\documentclass[leqno]{amsart}

\usepackage[numbers]{tsnatbib}
\usepackage{mathabx}
\usepackage{amsmath, amsthm}
\usepackage{amsfonts,amssymb,dsfont}
\usepackage{nicefrac,booktabs,mathrsfs}
\usepackage{bm}                                 
\usepackage{stmaryrd}    

\newcommand{\cA}{{\mathcal A}}  

\newcommand{\cB}{{\mathcal B}}

\newcommand{\ccF}{{\mathscr F}}\newcommand{\cF}{{\mathcal F}}
\newcommand{\ccG}{{\mathscr G}}\newcommand{\cG}{{\mathcal G}}
\newcommand{\ccH}{{\mathscr H}}\newcommand{\cH}{{\mathcal H}}

\newcommand{\cU}{{\mathcal U}}

\newcommand{\cX}{{\mathcal X}}

\newenvironment{enumeratei}
  {\begin{enumerate} }
  {\end{enumerate}}

\newcommand{\half}{\frac{1}{2}}

\newcommand{\Ind}{{\mathds 1}}
\newcommand{\ind}[1]{\Ind_{\{#1\}}}

\newcommand{\R}{\mathbb{R}}

\newcommand{\bbF}{\mathbb{F}}
\newcommand{\bbH}{\mathbb{H}}
\newcommand{\bbG}{\mathbb{G}}

\newtheorem{theorem}{Theorem}[section]
\newtheorem{lemma}[theorem]{Lemma}              
\newtheorem{proposition}[theorem]{Proposition}  

\theoremstyle{definition}
\newtheorem{example}{Example}[section]

\newtheorem{remark}{Remark}[section]

\newtheorem{assumption}{Assumption}[section]

\renewcommand{\cF}{\ccF}
\renewcommand{\cG}{\ccG}
\renewcommand{\cH}{\ccH} 

\usepackage[pdftex,colorlinks,urlcolor=red,citecolor=blue,linkcolor=red]{hyperref}

\renewcommand{\cF}{\ccF}
\renewcommand{\cG}{\ccG}
\renewcommand{\cH}{\ccH}

\begin{document}

\title[Generalized intensity]{A generalized intensity based framework for single-name credit risk
}
\author{Frank Gehmlich and Thorsten Schmidt}
\address{Frank Gehmlich email: gehmlichfrank@gmail.com and 
Thorsten Schmidt, University of Freiburg, Dep. of Mathematics, Eckerstr 1, 79106 Freiburg, Germany email:
thorsten.schmidt@stochastik.uni-freiburg.de}
%
%
 \date{November 19, 2015}

\maketitle

\begin{abstract}
The intensity of a default time is obtained by assuming that the default indicator process has an absolutely continuous compensator. Here we drop the assumption of absolute continuity with respect to the Lebesgue measure and only assume that the compensator is absolutely continuous with respect to a general $\sigma$-finite measure. This allows for example to incorporate the Merton-model in the generalized intensity based framework. An extension of the Black-Cox model is also considered.
We propose a class of generalized Merton models and study absence of arbitrage by a suitable modification of the forward rate approach of Heath-Jarrow-Morton (1992). Finally, we study affine term structure models which fit in this class. They exhibit stochastic discontinuities in contrast to the affine models previously studied in the literature.
\end{abstract}

\section{Introduction}


The two most common approaches to credit risk modelling are the \emph{structural} approach, pioneered in the seminal work of Merton \cite{Merton1974}, and the \emph{reduced-form} approach which can be traced back to early works of Jarrow, Lando, and Turnbull \cite{JarrowTurnbull1995,Lando94} and to \cite{ArtznerDelbaen95}.

The default of a company happens when the company is not able to meet its obligations. In many cases the debt structure of a company is known to the public, such that default happens with positive probability at times which are known a priori. This, however, is excluded in the intensity-based framework and it is the purpose of this article to put forward a generalisation which allows to incorporate such effects. Examples in the literature are, e.g., structural models like  \cite{Merton1974} and \cite{Geske1977,GeskeJohnson84}. The recently missed coupon payment by Argentina is an example for such a credit event as well as the default of Greece on the 1st of July\footnote{Argentina's missed coupon payment on \$29 billion debt was voted a credit event by the International Swaps and Derivatives Association, see the announcements in \cite{ISDAArgentina2014} and \cite{ReutersArgentina2014}. Regarding the failure of 1.5 Billon EUR of Greece on a scheduled debt repayment to the International Monetary fund, see e.g. \cite{NYTimes2015}.}.

It is a remarkable observation of \cite{BelangerShreveWong2004} that it is possible to extend the reduced-form approach beyond the class of intensity-based models. The authors study a class of first-passage time models under a filtration generated by a Brownian motion and show its use for pricing and modelling credit risky bonds. Our goal is to start with even weaker assumptions on the default time and to allow for jumps in the compensator of the default time at deterministic times. From this general viewpoint it turns out, surprisingly, that previously used HJM approaches lead to arbitrage: the whole term structure is absolutely continuous and can not compensate for points in time bearing a positive default probability. We propose a suitable extension with an additional  term allowing for discontinuities in the term structure at  certain random times and derive precise drift conditions for an appropriate  no-arbitrage condition. The related article \cite{GehmlichSchmidt2015} only allows for the special case of finitely many risky times, an assumption which is dropped in this article.

The structure of the article is as follows: in Section \ref{sec:tau}, we introduce the general setting and study drift conditions in an extended HJM-framework which guarantee absence of arbitrage in the bond market. In  Section \ref{secAffine} we study a class of affine models which are stochastically discontinuous. Section \ref{sec:conclusion} concludes.

\section{A general account on credit risky bond markets}\label{sec:tau}
 
Consider a filtered probability space $(\Omega, \cA, \bbG, P)$ with a filtration $\bbG=(\cG_t)_{t \ge 0}$ (the \emph{general} filtration) satisfying the usual conditions, i.e.\ it is right-continuous and  $\cG_0$ contains the $P$-nullsets $N_0$ of $\cA$. Throughout, the probability measure $P$ denotes the objective measure. As we use tools from stochastic analysis, all appearing filtrations shall satisfy the usual conditions. We follow the notation from \cite{JacodShiryaev} and refer to this work for details on stochastic processes which are not laid out here.

The filtration $\bbG$ contains all available information in the market. The default of a company is public information and we therefore assume that the default time $\tau$ is a $\bbG$-stopping time.
We denote the \emph{default indicator process} $H$ by  
  $$ H_t = \ind{t \ge \tau}, \qquad t \ge 0, $$
such that $H_t=\Ind_{\llbracket\tau,\infty\llbracket}(t)$ is a right-continuous, increasing process. We will also make use of the \emph{survival process} $1-H=\Ind_{\llbracket 0, \tau\llbracket}$. The following remark recalls the essentials of the well-known intensity based approach.

\begin{remark}[The intensity-based approach]
 The intensity-based approach consists in two steps: first, denote by $\bbH=(\cH_t)_{t \ge 0}$ the filtration generated by the default indicator, $\cH_t=\sigma(H_s:0 \le s \le t) \vee N_0$, and assume that there exists a sub-filtration $\bbF$ of $\bbG$, i.e.~$\cF_t \subset \cG_t$ holds for all $t \ge 0$ such that 
 \begin{align} \label{GFH} 
 \cG_t = \cF_t \vee \cH_t, \qquad t \ge 0. 
 \end{align} 
 Viewed from this perspective, $\bbG$ is obtained from the default information $\bbH$ by a \emph{progressive enlargement}\footnote{Note that here $\bbG$ is right-continuous and $P$-complete by assumption which is a priori not guaranteed by \eqref{GFH}. One can, however, use the right-continuous extension and we refer to \cite{GuoZeng2008} for a precise treatment and for a guide to the related literature.} with the filtration $\bbF$. This assumption opens the area for the largely developed field of enlargements of filtration with a lot of powerful and quite general results. 

 Second, the following key assumption specifies the default intensity: assume that there is an $\bbF$-progressive process $\lambda$, such that 
 \begin{align}
  P(\tau > t |\cF_t) = \exp\Big(-\int_0^t \lambda_s ds\Big), \quad t \ge 0.
 \end{align}
 It is immediate that the inclusion $\cF_t \subset \cG_t$ is strict under existence of an intensity, i.e.~$\tau$ is not a $\bbF$-stopping time. Arbitrage-free pricing can be achieved via the following result: Let $Y$ be a non-negative random variable. Then, for all $t \ge 0$,
 $$ E[\ind{\tau >t} Y |\cG_t] = \ind{\tau>t}e^{\int_0^t \lambda_s ds}E[\ind{\tau>t}Y|\cF_t]. $$
 Of course, this results hold also when a pricing measure $Q$ is used instead of $P$.
 For further literature and details we refer for example to \cite{Filipovic2009}, Chapter 12, and to \cite{BieleckiRutkowski2002}.
\end{remark}

\subsection{The generalized intensity-based framework}
The default indicator process $H$ is a bounded, c\'adl\'ag and increasing process, hence a submartingale of class (D), that is the family $(X_T)$ over all stopping times $T$ is uniformly integrable. By the Doob-Meyer decomposition\footnote{See \cite{KaratzasShreve1988}, Theorem 1.4.10.}, the process 
\begin{align}
\label{eq:M} M_t = H_t - \Lambda_t, \quad t \ge 0 
\end{align}
is a true martingale where $\Lambda$ denotes the dual $\bbF$-predictable projection, also called  compensator, of $H$. As $1$ is an absorbing state, $\Lambda_t=\Lambda_{t \wedge \tau}$. 
To keep the arising technical difficulties at a minimum, we assume that there is an increasing process $A$ such that 
\begin{align}\label{Hp}
\Lambda_t = \int_0^{t \wedge \tau} \lambda_s dA(s), \quad t \ge 0, 
\end{align}
with a non-negative and predictable process $\lambda$. The process $\lambda$ is called \emph{generalized intensity} and we refer to Chapter VIII.4 of \cite{Bremaud1981} for a more detailed treatment of generalized intensities (or, equivalently, dual predictable projections) in the context of point processes. 

Note that with $\Delta M \le 1$ we have that $\Delta \Lambda=\lambda_s\Delta A(s) \le 1$. 
Whenever $\lambda_s\Delta A(s) > 0$, there is a positive probability that the company defaults at time $s$. We call such times \emph{risky times}, i.e.\ predictable times having a positive probability of a default occurring right at that time. Note that under our assumption \eqref{Hp}, all risky times are deterministic. The relationship between $\Delta \Lambda(s)$ and the default probability at time $s$ will be clarified in Example \ref{ex:doublystochastic}.


\subsection{An extension of the HJM-approach}\label{sec:HJM}
A credit risky bond with maturity $T$ is a contingent claim promising to pay one unit of currency at  $T$. The price of the bond with maturity $T$ at time $t \le T$ is denoted by $P(t,T)$. If no default occurred prior to or at $T$ we have that $P(T,T)=1$. 
We will consider zero recovery, i.e.\ the bond loses its total value at default, such that $P(t,T)=0$ on $\{t \ge \tau\}$. The family of stochastic processes $\{(P(t,T)_{0 \le t \le T})$, $T\ge 0\}$ describes  the  evolution of the \emph{term structure} $T \mapsto P(.,T)$ over time.

Besides the bonds there is a \emph{num\'eraire} $X^0$, which is a strictly positive, adapted process. We make the  weak assumption that $\log X^0$ is absolutely continuous, i.e.~$X^0_t=\exp(\int_0^t r_s ds)$ with a progressively measurable process $r$, called the short-rate. For practical applications one would use the overnight index swap (OIS) rate for constructing such a num\'eraire.

The aim of the following is to extend the  HJM approach in an appropriate way to the generalized intensity-based framework in order to obtain arbitrage-free bond prices. First approaches in this direction were \cite{Schoenbucher:CRDerivatives,DuffieSingleton99} and a rich source of literature is again \cite{BieleckiRutkowski2002}.
Absence of arbitrage in such an infinite dimensional market can be described in terms of no asymptotic free lunch (NAFL) or the more economically meaningful no asymptotic free lunch with vanishing risk, see \cite{KleinSchmidtTeichmann2015} and \cite{CuchieroKleinTeichmann}.

 Consider a pricing measure $Q^*\sim P$. Our intention is to find conditions which render $Q^*$ an equivalent local martingale measure. 
In the following, only occasionally the measure $P$ will be used, such that from now on, all appearing terms (like martingales, almost sure properties, etc.) are to be considered with respect to $Q^*$.


To ensure that the subsequent analysis is meaningful, we make the  following technical assumption.
\begin{assumption}\label{ass1}
The generalized default intensity $\lambda$ is non-negative, predictable and $A$-integrable on $[0,T^*]$:
  $$ \int_0^{T^*} \lambda_s dA(s) < \infty, \quad Q^*\text{-a.s.} $$
Moreover, $A$ has vanishing singular part, i.e.
\begin{align}\label{ass:A} 
A(t) = t + \sum_{0<s\le t}\Delta A(s). 
\end{align}
\end{assumption}
The representation \eqref{ass:A} of $A$ is without loss of generality: indeed, if the continuous part $A^c$ is absolutely continuous, i.e. $A^c(t)=\int_0^t a(s) ds$, replacing $\lambda_s$ by $\lambda_sa(s)$ gives the compensator of $H$ with respect to $\tilde A$ whose continuous part is $t$.

Next, we aim at building an arbitrage-free framework for bond prices. In the generalized intensity-based framework, the (HJM) approach does allow for arbitrage opportunities at risky times. We therefore consider the following generalization: consider a $\sigma$-finite (deterministic) measure $\nu$. We could be  general on $\nu$,  allowing for an absolutely continuous, a singular continuous and a pure-jump part. However, for simplicity, we leave the singular continuous part aside and assume that
$$ \nu = \nu^{ac} + \nu^d $$
where $\nu^{ac}(ds)=ds$ and $\nu^d$ distributes mass only to points, i.e. $\nu^d(A)=\sum_{i \ge 1}w_i \delta_{u_i}(A)$, for $0<u_1<u_2<\dots$ and positive weights $w_i>0$, $i\ge 1$; here $\delta_u$ denotes the Dirac measure at $u$.
Moreover, we assume that defaultable bond prices are given by
  \begin{align}\label{PtT}
  P(t,T) &= \ind{\tau>t} \exp\bigg( -\int_t^T f(t,u) \nu(du)\bigg) \nonumber\\
  &= \ind{\tau>t} \exp\bigg( -\int_t^T f(t,u) du - \sum_{i \ge 1} \ind{u_i \in (t,T]} w_i f(t,u_i)\bigg) , \quad 0 \le t \le T \le T^*.
  \end{align}
The sum in the last line gives the extension over the (HJM) approach which allows us to deal with risky times in an arbitrage-free way.

The family of processes $(f(t,T))_{ 0 \le t \le T}$ for $T \in [0,T^*]$  are assumed to be It\^o processes satisfying
\begin{align}\label{f1}
f(t,T) &= f(0,T) + \int_0^t a(s,T)ds + \int_0^t b(s,T) \cdot dW_s
\end{align}
with an $n$-dimensional $Q^*$-Brownian motion $W$. 

Denote by $\cB$ the Borel $\sigma$-field over $\R$. 
\begin{assumption}\label{ass2}We require the following technical assumptions:
\begin{enumeratei}
\item  the initial forward curve is measurable, and integrable on $[0,T^*]$:
    $$\int_0^{T^*}|f(0,u)|<\infty, \qquad Q^*\text{-a.s.},$$
\item  the \emph{drift parameter} $a(\omega,s,t)$ is $\mathbb{R}$-valued $\mathcal{O}\otimes\mathcal{B}$-measurable and  integrable on $[0,T^*]$:
    $$\int_0^{T^*}\int_0^{T^*} |a(s,u)| ds\,\nu(du)<\infty, \quad Q^*\text{-a.s.},$$
\item  the \emph{volatility parameter} $b(\omega,s,t)$ is $\mathbb{R}^n$-valued, $\mathcal{O}\otimes\mathcal{B}$-measurable, and
    $$ \sup_{s,t\leq T^*} \parallel b(s,t) \parallel <\infty, \quad Q^*\text{-a.s.}$$
\item  it holds that
$$ 0 \le \lambda(u_i)\Delta A(u_i) < w_i, \quad i \ge 1. $$
\end{enumeratei}\end{assumption}

Set
\begin{align}\label{abisbetabar}
\begin{aligned}
\bar{a}(t,T) &= \int_t^T a(t,u) \nu(du), \\ 
\bar{b}(t,T) &= \int_t^T b(t,u) \nu(du), \\
H'(t) &= \int_0^t \lambda_s ds - \sum_{u_i \le t} \log\Big( \frac{w_i-\lambda_{u_i}\Delta A(u_i)}{w_i}\Big).
\end{aligned}
\end{align}
The following proposition gives the desired drift condition in the generalized Merton models. 

\begin{theorem}\label{prop:dcm}
Assume  that Assumptions \ref{ass1} and \ref{ass2} hold. Then $Q^*$ is an ELMM if and only if the following conditions hold:   $\{s:\Delta A(s) \neq 0\}\subset \{u_1,u_2,\dots\}$,  and
\begin{align}
\int_0^t f(s,s) \nu(ds) &= \int_0^t r_s ds + H'(t),  \label{dcm1}\\
\bar a(t,T)  &=  \half \parallel \bar b(t,T) \parallel^2 , \label{dcm2}
\end{align}
for $0 \le t \le T \le T^*$  $dQ^* \otimes dt$-almost surely on $\{t<\tau\}$.
\end{theorem}

The first condition, \eqref{dcm1}, can be split in the continuous and pure-jump part, such that \eqref{dcm1} is equivalent to
\begin{align*}
 f(t,t) &= r_s + \lambda_s \\
 f(t,u_i) &= \log \frac{w_i}{w_i-\lambda(u_i)\Delta A(u_i) }\ge 0.
\end{align*}
The second relation states explicitly the connection of the forward rate at a risky time $u_i$ to the probability $Q^*(\tau=u_i|\cF_{u_i-})$, given that $\tau\ge u_i$, of course. It simplifies moreover, if $\Delta A(u_i)=w_i$ to 
\begin{align}\label{temp244} 
f(t,u_i) &= -\log (1-\lambda(u_i)).
\end{align}

For the proof we first provide the canonical decomposition of 
$$ J(t,T):= \int_t^T f(t,u)\nu(du), \qquad 0 \le t \le T. $$
\begin{lemma}\label{lem32}
Assume that Assumption \ref{ass2} holds. Then, for each $T \in [0,T^*]$ the process
$(J(t,T))_{0 \le t \le T}$ is a special semimartingale and
\begin{align*}
J(t,T) &= \int_0^T f(0,u)\nu(du) + \int_0^t\bar a(u,T)du +\int_0^t\bar b(u,T)dW_u  -\int_0^t  f(u,u) \nu(du).
\end{align*} 
\end{lemma}
\begin{proof}
Using the stochastic Fubini Theorem (as in \cite{Veraar12}), we obtain
\begin{align*}
J(t,T)&= \int_t^T \bigg(f(0,u)+\int_0^t  a(s,u)ds + \int_0^t b(s,u)dW_s\bigg)\nu(du)\\
&= \int_0^Tf(0,u)\nu(du) +\int_0^t\int_s^T a(s,u)\nu(du)ds + \int_0^t\int_s^T b(s,u)\nu(du)dW_s\\
&- \int_0^tf(0,u)\nu(du) -\int_0^t\int_s^t a(s,u)\nu(du)ds - \int_0^t\int_s^t b(s,u)\nu(du)dW_s\\
&= \int_0^Tf(0,u)\nu(du) +\int_0^t\bar a(s,T)ds + \int_0^t\bar b(s,T)dW_s\\
&- \int_0^t\bigg(f(0,u) -\int_0^u a(s,u)ds - \int_0^u b(s,u)dW_s\bigg)\nu(du),
\end{align*}
and the claim follows.
\end{proof}
\begin{proof}[Proof of Theorem \ref{prop:dcm}]
Set, $E(t) = \ind{\tau>t}$, and  $F(t,T) = \exp\Big(-\int_t^T f(t,u) \nu(du) \Big)$, such that
$ P(t,T)=E(t)F(t,T)$. Integration by parts yields that
\begin{align}\label{PMtTrepr}
  dP(t,T) &=  F(t-,T) d E(t) + E(t-) dF(t,T) + d [E, F(.,T) ]_t  =: (1')+(2')+(3'). 
\end{align}
In view of (1$'$), we obtain from  \eqref{Hp},   that
\begin{align}\label{defM2}
E(t) + \int_0^{t \wedge \tau } \lambda_s dA(s)=: M^1_t  
\end{align}
is a martingale. 
Regarding (2$'$), note that from Lemma \ref{lem32} we obtain by It\^o's formula that
\begin{align}\label{temp563}
  \begin{aligned}
  \frac{d F(t,T)}{F(t-,T)} &= \Big( f(t,t) - \bar a(t,T) + \half \parallel \bar b(t,T) \parallel^2 \Big) dt \\
  &+ \sum_{i \ge 0}  \left(e^{f(t,t)}-1\right) w_i\delta_{u_i}(dt) + dM^2_t,
  \end{aligned}
\end{align}
with a local martingale $M^2$. For the remaining term (3$'$), note that
\begin{align}\label{temp571}
  \sum_{0 < s \le t } \Delta E(s) \Delta F(s,T) &=  \int_0^t F(s-,T) (e^{f(s,s)}-1) \nu(\{s\}) dE(s) \\
  &= \int_0^t F(s-,T) (e^{f(s,s)}-1) \nu(\{s\}) d M^1_s \nonumber \\
  &- \int_0^{t \wedge \tau}    F(s-,T) (e^{f(s,s)}-1) \nu(\{s\}) \lambda_s dA(s). \nonumber
\end{align}
Inserting  \eqref{temp563} and \eqref{temp571} into \eqref{PMtTrepr} we obtain 
\begin{align*}
\frac{dP(t,T)}{P(t-,T)} &= -\lambda_t  dA(t) \\
   & +  \Big( f(t,t) -\bar a(t,T) +  \half \parallel \bar b(t,T) \parallel^2  \Big) dt \\
   & + \sum_{i \ge 0}  \left(e^{f(t,t)}-1\right) w_i\delta_{u_i}(dt)\\
   & - \int_{\R}  \nu(\{t\})(e^{f(t,t)}-1)  \lambda_t dA(t)+ dM^3_t
\end{align*}
with a local martingale $M^3$. We obtain a $Q^*$-local martingale if and only if the drift vanishes. Next, we can separate between absolutely continuous and discrete part. The absolutely continuous part yields \eqref{dcm2} and $f(t,t)=r_t+\lambda_t$ $dQ^*\otimes dt$-almost surely. It remains to compute the discontinuous part, which is given  by
\begin{align*}
\sum_{i:u_i \le t} P(u_i-,T) (e^{f(u_i,u_i)}-1) w_i-  \sum_{0<s\le t} P(s-,T)  e^{f(s,s)} \lambda_s \Delta A(s),
\end{align*}
for $0 \le t \le T \le T^*$. This yields $\{s:\Delta A(s)\neq 0\}\subset\{u_1,u_2,\dots\}$.
The discontinuous part vanishes if and only if 
\begin{align*} 
  \ind{u_i \le T^*\wedge \tau}e^{- f(u_i,u_i)} w_i= & \ind{u_i \le T^* \wedge \tau} \Big(w_i-\lambda_{u_i}\Delta A(u_i)\Big),\qquad i \ge 1,
\end{align*}
which is equivalent to 
\begin{align*} 
  \ind{u_i \le T^*\wedge \tau} f(u_i,u_i) = & -\ind{u_i \le T^* \wedge \tau} \log \frac{w_i-\lambda_{u_i}\Delta A(u_i)}{w_i},\qquad i \ge 1.
\end{align*}
We obtain  \eqref{dcm1} and  the claim follows.
\end{proof}

\begin{example}[The Merton model]\label{exp:merton_model}
The paper \cite{Merton1974}  considers a simple capital structure of a firm, consisting only of equity and a zero-coupon bond with maturity $U>0$. The firm defaults at $U$ if the total market value of its assets is not sufficient to cover the liabilities. 

We are interested in setting up an arbitrage-free market for credit derivatives and consider a market of defaultable bonds $P(t,T)$, $0 \le t \le T \le T^*$ with $0< U \le T^*$ as basis for more complex derivatives. In a stylized form the Merton model can be represented by a Brownian motion $W$ denoting the normalized logarithm of the firm's assets, a constant $K>0$ and the default time
$$ \tau = \begin{cases}
U      & \text{if }W_U \le K \\
\infty & \text{otherwise}.
\end{cases}$$ 
Assume for simplicity a  constant interest rate $r$ and let $\bbF$ be the filtration generated by $W$.  Then $P(t,T) =e^{-r(T-t)}$ whenever $T<U$ because these bonds do not carry default risk. On the other hand, for $t< U \le T $, 
\begin{align*}
P(t,T) &= e^{-r(T-t)}E^*[\ind{\tau>T}|\cF_t]= e^{-r(T-t)} E^*[\ind{\tau=\infty}|\cF_t]= e^{-r(T-t)}\Phi\bigg(\frac{W_t - K}{\sqrt{U-t}}\bigg),
\end{align*}
where $\Phi$ denotes the cumulative distribution function of a standard normal random variable and $E^*$ denotes the expectation with respect to $Q^*$. For $t\to U$ we recover $P(U,U)=\ind{\tau=\infty}$. The derivation of representation \eqref{PtT} with $\nu(du):=du+\delta_{U}(du)$ is straightforward. A simple calculation with 
\begin{align}\label{merton:PtT}
P(t,T) &= \ind{\tau > t} \exp\bigg( -\int_t^T f(t,u) du - f(t,U)\ind {t < U \le T} \bigg)
\end{align}
yields $f(t,T)=r$ for $T \not = U$ and 
$$ f(t,U) =- \log \Phi\bigg(\frac{W_t - K}{\sqrt{U-t}}\bigg). $$
By It\^{o}'s formula we obtain 
\begin{align*}
b(t,U)&= - \frac{\varphi\bigg(\frac{W_t - K}{\sqrt{U-t}}\bigg)}{\Phi\bigg(\frac{W_t - K}{\sqrt{U-t}}\bigg)}(U-t)^{-1/2},
\end{align*}
and indeed, $a(t,U)= \frac{1}{2}b^2(t,U).$ Note that  the conditions for Proposition \ref{prop:dcm} hold and, the market consisting of the bonds $P(t,T)$ satisfies NAFL, as expected. More flexible models of arbitrage-free bond prices can be obtained if the market filtration $\bbF$ is allowed to be more general, as we show in Section \ref{secAffine} on affine generalized Merton models.
\end{example}

\begin{example}(An extension of the Black-Cox model)
The model suggested in \cite{BlackCox1976} uses a first-passage time approach to model credit risk. Default happens at the first time, when the firm value falls below a pre-specified boundary, the default boundary. We consider a stylized version of this approach and continue the example \ref{exp:merton_model}. Extending the original approach, we include a zero-coupon bond with maturity $U$. The reduction of the firm value at $U$ is equivalent to considering a default boundary with an upward jump at that time. Hence, we consider a Brownian motion $W$ and the default boundary
$$ D(t) = D(0) + K \ind{ U \ge t}, \qquad t \ge 0, $$
with $D(0)<0$, and let default be the first time when $W$ hits $D$, i.e.
$$ \tau = \inf\{ t \ge 0: W_t \le D(t) \} $$
with the usual convention that $\inf\,  \emptyset = \infty$. 
The following lemma computes the default probability in this setting and the forward rates are directly obtained from this result together with \eqref{merton:PtT}.  The filtration $\bbG=\bbF$ is given by the natural filtration of the Brownian motion $W$ after completion. Denote the random sets
\begin{align*}
\Delta_1& :=\{(x,y)\in\R^2:x\sqrt{T-U}  \le D(U) - (y\sqrt{U-t} + W_t), y\sqrt{U-t} + W_t > D(0) \}\\
\Delta_2& :=\{(x,y)\in\R^2:x\sqrt{T-U}  \le D(U) - (y\sqrt{U-t} + 2D(0)-W_t),\\
&\phantom{\{(x,y)\in\R^2:\sqrt{T-U} x \le D(U)} y \sqrt{U-t} + D(0)-W_t > 0\} .
\end{align*}
\begin{lemma}
Let $D(0)<0$, $U>0$ and $D(U) \ge D(0)$.   For $0 \le t < U$, it holds  on $\{\tau >t \}$, that
\begin{align} \label{temp409} 
P(\tau > T|\cF_t) = 1-2\Phi\Big( \frac{D(0)-W_t}{\sqrt{T-t}} \Big) -\ind{ T \ge U} 2(\Phi_2(\Delta_1)-\Phi_2(\Delta_2)) ,
\end{align}
where $\Phi_2$ is the distribution of a two-dimensional standard normal distribution and the sets $\Delta_t=\Delta_t(D),\  t \ge U$ are given by
$$ \Delta_t= \{(x,y) \in \R^2: x\sqrt{T-U} + y \sqrt{U} \le -D(U),  \}. $$ For $t \ge U$ it holds  on $\{\tau >t \}$, that
$$ P(\tau > T|\cF_t) = 
  1-2\Phi\bigg( \frac{D(U)-W_t}{\sqrt{T-t}} \bigg). $$
\end{lemma}
\begin{proof}
The first part  of \eqref{temp409} where $T < U$ follows directly from the reflection principle and the property that $W$ has independent and stationary increments. Next, consider $0 \le t < U \le T$. Then, on $\{W_U > D(U)\}$,
\begin{align}\label{temp422}
P(\inf_{[U,T]}W > D(U) | \cF_U) &= 1-2\Phi\bigg( \frac{D(U)-W_U}{\sqrt{T-U}} \bigg).
\end{align}
Moreover, on $\{W_t>D(0)\}$ it holds for $x>D(0)$ that
\begin{align*}
P(\inf_{[0,U]} W > D(0), W_U>x|\cF_t) &= P(W_U>x|\cF_t) - P(W_U<x, \inf_{[0,U]} W \le D(0)|\cF_t) \\
&= \Phi\bigg( \frac{W_t-x}{\sqrt{U-t}} \bigg) - \Phi\bigg( \frac{2 D(0)-x-W_t}{\sqrt{U-t}} \bigg).
\end{align*}
Hence, $E[g(W_U) \ind{\inf_{[0,U]} W > D(0)}|\cF_t] = \ind{\inf_{[0,t]} W > D(0)}\int_{D(0)}^\infty g(x)f_t(x) dx$ with  density 
$$ f_t(x) = \ind{x>D(0)} \frac{1}{\sqrt{U-t}} \Big[ \phi\Big(\frac{W_t-x}{\sqrt{U-t}}\Big) 
- \phi\Big( \frac{2D(0)-x-W_t}{\sqrt{U-t}} \Big) \Big]. $$ 
Together with \eqref{temp422} this yields on $\{\inf_{[0,t]} W > D(0)\}$
\begin{align*}
   P(\inf_{[0,T]}(W-D) > 0 | \cF_t) &= \int_{D(0)}^\infty \bigg[1-2\Phi\bigg( \frac{D(U)-x}{\sqrt{T-U}}\bigg)\bigg] f_t(x) dx \\
   &= P(\inf_{[t,T]} W > D(0) |\cF_t)- 2 \int_{D(0)}^\infty \Phi\bigg( \frac{D(U)-x}{\sqrt{T-U}}\bigg) f_t(x) dx.
\end{align*}
It remains to compute the integral. Regarding the first part, letting $\xi$ and $\eta$ be independent and standard normal, we obtain that
\begin{align*}
\hspace{1cm}&\hspace{-1cm}\int_{D(0)}^\infty \Phi\bigg( \frac{D(U)-x}{\sqrt{T-U}}\bigg) \frac{1}{\sqrt{U-t}} \phi\Big(\frac{x-W_t}{\sqrt{U-t}}\Big)  dx \\
&= P_t\Big( \sqrt{T-U} \xi \le D(U) - (\sqrt{U-t}\eta + W_t), \sqrt{U-t}\eta + W_t > D(0) \Big) \\
&= \Phi_2(\Delta_1),
\end{align*}
where we abbreviate $P_t(\cdot)=P(\cdot|\cF_t)$.
In a similar way,
\begin{align*}
&\int_{D(0)}^\infty \Phi\bigg( \frac{D(U)-x}{\sqrt{T-U}}\bigg) \frac{1}{\sqrt{U-t}}\phi\Big( \frac{x-(2D(0)-W_t)}{\sqrt{U-t}} \Big)  dx \\
&=P_t\Big( \sqrt{T-U} \xi \le D(U) - (\sqrt{U-t}\eta + 2D(0)-W_t), \sqrt{U-t}\eta + D(0)-W_t > 0 \Big) \\
&= \Phi_2(\Delta_2)
\end{align*}
and we conclude.
\end{proof}

\end{example}

\section{Affine models in the generalized intensity-based framework}\label{secAffine}
Affine processes are a well-known tool in the financial literature and one reason for this is their analytical tractability. In this section we closely follow \cite{GehmlichSchmidt2015} and shortly state the appropriate affine models which fit the generalized intensity framework. For proofs, we refer the reader to this paper. 

The main point is that affine processes  in the literature are assumed to be \emph{stochastically continuous} (see \cite{DuffieFilipovicSchachermayer} and \cite{Filipovic05}).  Due to the discontinuities introduced in the generalized intensity-based framework,  we propose to consider \emph{piecewise continuous affine processes}.

\begin{example}\label{ex:doublystochastic}
Consider a non-negative integrable function $\lambda$, a constant $\lambda'\ge 0$ and a deterministic time $u>0$. Set
$$ K(t)= \int_0^t \lambda(s) ds + \ind{t \ge u} \kappa, \qquad t \ge 0. $$ 
Let the default time $\tau$ be given by $\tau=\inf\{t \ge 0: K_t \ge \zeta\}$  with a standard exponential-random variable $\zeta$. Then $P(\tau = u)= 1-e^{-\kappa}=:\lambda'$. Considering $\nu(ds)=ds + \delta_u(ds)$ with $u_1=u$ and $w_1=1$, we are in the setup of the previous section. The drift condition \eqref{dcm1} holds, if 
$$ f(u,u) = - \log(1-\lambda') = \kappa.$$
Note, however, that $K$ is not the compensator of $H$. Indeed, the compensator of $H$ equals $\Lambda_t=\int_0^{t \wedge \tau} \lambda(s) ds + \ind{t \ge u} \lambda'$, see \cite{JeanblancRutkowski2000} for general results in this direction.\end{example}

The purpose of this section is to give a suitable extension of the above example  involving  affine processes. 
 Recall that we consider a $\sigma$-finite measure  
 $$ \nu(du) = du + \sum_{i\ge 1} w_i\delta_{u_i}(du), $$
 as well as $A(u) = u + \sum_{i \ge 1} \ind{u \ge u_i }$. The idea is to consider an affine process $X$ and study arbitrage-free doubly stochastic term structure models where the compensator $\Lambda$ of the default indicator process $H=\ind{\cdot \le \tau}$ is given by
\begin{align} \label{affineHp} 
\Lambda_t = \int_0^t \Big( \phi_0(s)+\psi_0(s)^\top \cdot X_s \Big) ds + \sum_{i\ge 1}\ind{t \ge u_i} \Big(1-e^{- \phi_i - \psi_i^\top \cdot X_{u_i}}\Big).
\end{align}
Note that by continuity of $X$, $\Lambda_t(\omega)<\infty$ for almost all $\omega$.
To ensure that $\Lambda$ is non-decreasing we will require that $\phi_0(s)+\psi_0(s)^\top \cdot X_s  \ge 0$ for all $s \ge 0$ and 
$\phi_i + \psi_i^\top \cdot X_{u_i} \ge 0$ for all $i\ge 1$. 

Consider a state space in canonical form $\cX=\R_{\ge 0}^m \times \R^n$ for integers $m,n \ge 0$ with $m+n=d$ and a $d$-dimensional Brownian motion $W$. Let $\mu$ and $\sigma$ be defined on $\cX$ by
\begin{align}
\label{affine_mu}
\mu(x) &= \mu_0 + \sum_{i=1}^dx_i\mu_i,\\
\label{affine_sigma}
\half \sigma(x)^\top\sigma(x) &= \sigma_0 + \sum_{i=1}^dx_i\sigma_i,
\end{align}
where $\mu_0,\mu_i\in\R^d$, $\sigma_0,\sigma_i\in\R^{d\times d}$, for all $i\in\{1,\ldots,d\}$. We assume that the parameters $\mu^i,\ \sigma^i$, $i=0,\dots,d$ are admissible in the sense of Theorem 10.2 in \cite{Filipovic2009}. Then the continuous, unique strong solution  of the stochastic differential equation
\begin{align}
dX_t &= \mu(X_t)dt + \sigma(X_t)dW_t,\quad X_0=x, 
\end{align}
is an \emph{affine} process $X$ on the state space $\cX$, see Chapter 10 in \cite{Filipovic2009} for a detailed exposition.

We call a bond-price model \emph{affine} if there exist functions $A:\R_{\ge 0}\times \R_{\ge 0} \to \R$, $B:\R_{\ge 0}\times \R_{\ge 0} \to \R^d$ such that
\begin{align}\label{affineMertonmodel}
P(t,T)= \ind{\tau > t} e^{-A(t,T)-B(t,T)^\top \cdot X_t},
\end{align}
for $0 \le t \le T \le T^*$. We assume that $A(.,T)$ and $B(.,T)$ are right-continuous. Moreover, we assume that  $t \mapsto A(t,.)$ and $t \mapsto B(t,.)$ are differentiable from the right and denote by $\partial_t^+$ the right derivative.  
For the convenience of the reader we state  the following proposition giving sufficient conditions for absence of arbitrage in an affine generalized-intensity based setting. It extends  \cite{GehmlichSchmidt2015} where only finitely many risky times were treated.

\begin{proposition}\label{prop:affine}
Assume that $\phi_0:\R_{\ge 0}\to \R$, $\psi_0:\R_{\ge 0}\to \R^d$ are continuous,  $\psi_0(s)+\psi_0(s)^\top \cdot x \ge 0$ for all $s \ge 0$ and $x \in \cX$ and the constants $\phi_i \in \R$ and $\psi_i \in \R^d$, $i \ge 1$ satisfy $\phi_i + \psi_i^\top \cdot x \ge 0$ for all $1 \le i \le n$ and $x \in \cX$ as well as $\sum_{i \ge 1}|w_i|(|\phi_i| +|\psi_{i,1}|+\dots+|\psi_{i,d}|)<\infty$.
Moreover, let the functions $A:\R_{\ge 0}\times \R_{\ge 0} \to \R$ and $B:\R_{\ge 0}\times \R_{\ge 0}\to \R^d$ be the unique solutions of 
\begin{align} \label{ric1}
\begin{aligned}A(T,T) &= 0 \\
A(u_i,T) &= A(u_i-,T)-\phi_i w_i\\
- \partial_t^+ A(t,T) &= \phi_0(t) +  \mu_0^\top \cdot B(t,T) 
- B(t,T)^\top \cdot \sigma_0 \cdot B(t,T),  \\
\end{aligned} \intertext{and} \label{ric2}
\begin{aligned}
B(T,T) &= 0 \\
B_k(u_i,T) &= B_k(u_i-,T) - \psi_{i,k} w_i\\ 
-\partial_t^+ B_k(t,T) &= \psi_{0,k}(t) +  \mu_k^\top \cdot B(t,T)  - B(t,T) ^\top \cdot \sigma_k \cdot B(t,T),
\end{aligned}
\end{align}
for $0 \le t \le T$.
Then, the doubly-stochastic affine model given by \eqref{affineHp} and \eqref{affineMertonmodel} satisfies NAFL. 
\end{proposition}

\begin{proof}
By construction,
\begin{align*}
A(t,T)  = \int_t^T a'(t,u) du+\sum_{i:u_i \in (t,T]} \phi_i w_i \\
B(t,T)  = \int_t^T b'(t,u) du+\sum_{i:u_i \in (t,T]} \psi_i w_i
\end{align*}
with suitable functions $a'$ and $b'$ and $a'(t,t)=\phi_0(t)$ as well as $b'(t,t)=\psi_0(t)$.
A comparison of \eqref{affineMertonmodel} with \eqref{PtT} yields the following: on the one hand, for $T =u_i\in \cU$, we obtain $f(t,u_i)=\phi_i+\psi_i^\top \cdot X_t$.  Hence, the coefficients $a(t,T)$ and $b(t,T)$ in \eqref{f1} for $T=u_i \in \cU$ compute to $a(t,u_i)=\psi_i^\top \cdot \mu(X_t)$ and $b(t,u_i) = \psi_i^\top \cdot \sigma(X_t)$.

On the other hand, for $T \not \in \cU$ we obtain that $f(t,T) = a'(t,T)+b'(t,T)^\top \cdot X_t$. Then, the coefficients $a(t,T)$ and $b(t,T)$  can be computed  as follows: applying It\^o's formula to $f(t,T)$ and comparing with \eqref{f1} yields that
\begin{align} \label{abaraffine}
\begin{aligned}
  a(t,T) &= \partial_t a'(t,T) + \partial_t b'(t,T)^\top \cdot X_t + b'(t,T)^\top \cdot \mu(X_t) \\
  b(t,T) &= b'(t,T)^\top \cdot \sigma(X_t).
\end{aligned}
\end{align}
Set $\bar a'(t,T)=\int_t^T a'(t,u) du$ and $\bar b'(t,T)=\int_t^T b'(t,u) du$ and note that, 

$$ \int_t^T \partial_t a'(t,u) du = \partial_t \bar a'(t,T) + a'(t,t). $$
As $\partial_t^+ A(t,T)=\partial_t \bar a'(t,T)$, and $\partial_t^+ B(t,T)=\partial_t \bar b'(t,T)$, we obtain from \eqref{abaraffine} that 
\begin{align*}
  \bar a(t,T) & = \int_t^T a (t,u) \nu (du) = \int_t^T a(t,u) du  + \sum_{u_i \in (t,T]} w_i\psi_i^\top \cdot \mu(X_t) \\
  &= \partial_t^+ A(t,T) + a'(t,t) 
  + \big( \partial_t^+ B(t,T) + b'(t,t) \big)^\top \cdot X_t + B(t,T)^\top  \cdot \mu(X_t), \\ 
   \bar b(t,T) & =\int_t^T b(t,u) \nu(du) = \int_t^T b(t,u) du +  \sum_{u_i \in (t,T]} w_i\psi_i^\top \cdot \sigma(X_t) \\
   &=B(t,T)^\top \cdot \sigma(X_t)
\end{align*}
for $0 \le t \le T \le T^*$.
We now show that under our assumptions, the drift conditions \eqref{dcm1}-\eqref{dcm2} hold: 
Observe that, by equations \eqref{ric1}, \eqref{ric2}, and the affine specification \eqref{affine_mu}, and \eqref{affine_sigma}, 
the drift condition  \eqref{dcm2} holds. 
Moreover, from \eqref{temp244},
$$ \Delta H'(u_i) =  \phi_i + \psi_i^\top \cdot X_{u_i} $$
and $\lambda_s= \phi_0(s) + \psi_0(s)^\top \cdot X_s$ by \eqref{affineHp}. We recover $\Delta \Lambda_{u_i} = 1-\exp(-\phi_i - \psi_i^\top \cdot X_{u_i})$ taking values in $[0,1)$ by assumption. Hence, 
 \eqref{dcm1} holds and the claim follows. 
\end{proof}

\begin{example}
In the one-dimensional case we consider  $X$,  given as solution of 
\begin{align*}
dX_t &= (\mu_0+\mu_1 X_t)dt + \sigma\sqrt{X_t}dW_t, \qquad t \ge 0.
\end{align*}
Consider only one risky time $u_1=1$ and let $\phi_0=\phi_1 = 0$, $\psi_0=1$, such that
$$ \Lambda = \int_0^t X_s ds + \ind{u \ge 1} (1-e^{-\psi_1 X_{1}}). $$
Hence the probability of having no default at time $1$ just prior to $1$ is given by $e^{-\psi_1 X_{1}}$, compare Example \ref{ex:doublystochastic}.

An arbitrage-free model can be obtained by choosing $A$ and $B$ according to Proposition \ref{prop:affine} which can be immediately achieved using Lemma 10.12 from \cite{Filipovic2009} (see in particular Section 10.3.2.2 on the CIR short-rate model): denote $\theta=\sqrt{\mu_1^2+2\sigma^2}$ and
\begin{align*}
L_1(t) &= 2(e^{\theta t}-1) ,\\
L_2(t) &= \theta(e^{\theta t}+1) + \mu_1 (e^{\theta t}-1) ,\\
L_3(t) &= \theta(e^{\theta t}+1) - \mu_1 (e^{\theta t}-1) ,\\
L_4(t) &= \sigma^2(e^{\theta t}-1).
\end{align*}
Then 
\begin{align*} 
  A_0(s) &= \frac{2 \mu_0}{\sigma^2} \log \Big( \frac{2 \theta e^{\frac{(\sigma - \mu_1)t}{2}} } {L_3(t)}\Big), \quad B_0(s) = -\frac{L_1(t)}{L_3(t)} 
\end{align*}
are the unique solutions of the Riccati equations $B_0'=\sigma^2 B_0^2-\mu_1 B_0 $ with boundary condition $B_0(0)=0$ and $A_0'=-\mu_0 B_0$ with boundary condition $A_0(0)=0$. Note that with $A(t,T) = A_0(T-t)$ and $B(t,T)= B_0(T-t)$ for $0 \le t \le T < 1$, the conditions of Proposition \ref{prop:affine} hold. Similarly,
for $1 \le t \le T$, choosing
$A(t,T)=A_0(T-t)$ and $B(t,T)=B_0(T-t)$ implies again the validity of \eqref{ric1} and \eqref{ric2}. On the other hand, for $0 \le t <1$ and $T \ge 1$ we set $u(T)=B(1,T)+\psi_1=B_0(T-1)+\psi_1$, according to \eqref{ric2}, and let 
\begin{align*}
  A(t,T) &= \frac{2 \mu_0}{\sigma^2} \log \Big( \frac{2 \theta e^{\frac{(\sigma - \mu_1)(1-t)}{2}} } {L_3(1-t)-L_4(1-t)u(T)}\Big)\\
  B(t,T) &= -\frac{L_1(1-t)-L_2(1-t)u(T)}{L_3(1-t)-L_4(1-t)u(T)}.
\end{align*}
It is easy to see that \eqref{ric1} and \eqref{ric2} are also satisfied in this case, in particular  $\Delta A(1,T)=-\phi_1=0$ and $\Delta B(1,T)=-\psi_1$. Note that, while $X$ is continuous, the bond prices are not even stochastically continuous because they jump almost surely at $u_1=1$.  We conclude by Proposition \ref{prop:affine} that this affine model is arbitrage-free. \hfill $\diamond$
\end{example}

\section{Conclusion}\label{sec:conclusion}

In this article we studied a new class of dynamic term structure models with credit risk where the compensator of the default time may jump at predictable times. This framework was called generalized intensity-based framework. It extends existing theory and allows to include  Merton's model, in a reduced-form model for pricing credit derivatives. Finally,  we studied a class of highly tractable affine models which are  only piecewise stochastically continuous.



\end{document}